\documentclass[submission,copyright,creativecommons]{eptcs}
\providecommand{\event}{LSFA 2026} 

\usepackage{iftex}
\usepackage{graphicx}
\usepackage{hyperref}
\usepackage{mathpartir}
\usepackage{amssymb}
\usepackage{amsmath}
\usepackage{xcolor}
\usepackage{stmaryrd}
\usepackage{xspace}
\usepackage{color}

\usepackage{amsthm}
\newtheorem{definition}{Definition}
\newtheorem{theorem}{Theorem}

\newtheorem{example}{Example}
\ifpdf
  \usepackage{underscore}         
  \usepackage[T1]{fontenc}        
\else
  \usepackage{breakurl}           
\fi

\title{Multiparty Session Types for GDPR Purpose Compliance}
\author{Evangelia Vanezi\qquad Dimitrios Kouzapas \qquad Anna Philippou
\institute{Department of Computer Science, University of Cyprus, Nicosia, Cyprus}
\email{\{vanezi.evangelia,kouzapas.dimitrios,philippou.anna\}@ucy.ac.cy}
}
\def\titlerunning{MPST for GDPR Purpose Compliance}
\def\authorrunning{E. Vanezi, D. Kouzapas \& A. Philippou}

\usepackage{ifthen}

\DeclareMathAlphabet{\pazocal}{OMS}{zplm}{m}{n}

\newcommand{\LComm}{\lts{Comm}}
\newcommand{\LBranch}{\lts{BranchSel}}

\newcommand{\SUBINT}{\lts{SubInt}}
\newcommand{\SUBBOOL}{\lts{SubBool}}
\newcommand{\SUBDATA}{\lts{SubPD}}
\newcommand{\SUBEND}{\lts{SubEnd}}

\newcommand{\SUBINP}{\lts{SubInp}}
\newcommand{\SUBOUT}{\lts{SubOut}}
\newcommand{\SUBPDOUT}{\lts{SubPDOut}}
\newcommand{\SUBPDINP}{\lts{SubPDInp}}
\newcommand{\SUBSEL}{\lts{SubSel}}
\newcommand{\SUBBRA}{\lts{SubBranch}}

\newcommand{\mymath}[1]{\ensuremath{#1}\xspace}

\newcommand{\set}[1]{\mymath{\{#1\}}}
\newcommand{\setbar}[1]{\mymath{\ |\ }}
\newcommand{\bnfis}{\mymath{\ \ ::=\ \ }}
\newcommand{\bnfbar}{\mymath{\ \ |\ \ }}

\newcommand{\tree}[3][]{
	\ensuremath{
        \displaystyle
    	\ifthenelse{\equal{#1}{}}{}{#1}
            \frac{
        	#2
            }{
        	#3
            }
    }
}

\newcommand{\dfracdouble}[2]{%
  \genfrac{}{}{}{}{#1}{#2}\mkern-12mu
  \raisebox{0.6ex}{\rule{1.2em}{0.1pt}}
  \raisebox{-0.6ex}{\rule{1.2em}{0.1pt}}
}

\newcommand{\treeDouble}[3][]{
	\ensuremath{
        \displaystyle
    	\ifthenelse{\equal{#1}{}}{}{#1}
            \dfracdouble{
        	#2
            }{
        	#3
            }
    }
}

\newcommand{\dtree}[3][] {
    {
        \mprset{fraction={===}}
        \inferrule*[left=#1] {
            #2
        }{
            #3
        }
    }
}

\newcommand{\sep}{.}
\newcommand{\outsymb}{!}
\newcommand{\inpsymb}{?}
\newcommand{\selsymb}{\triangleleft}
\newcommand{\brasymb}{\triangleright}

\newcommand{\const}[1][]{\mymath{\mathsf{c}_{#1}}}
\newcommand{\constd}[1][]{\mymath{\mathsf{c'}_{#1}}}

\newcommand{\PP}[1][]{\mymath{P_{#1}}}

\newcommand{\outp}[2]{\mymath{#1\outsymb\langle#2\rangle\sep}}
\newcommand{\inpp}[2]{\mymath{#1\inpsymb(#2)\sep}}
\newcommand{\selp}[2]{\mymath{#1\selsymb#2\sep}}

\newcommand{\brapi}[3]{\mymath{#1\brasymb\{#2\}_{#3}}}
\newcommand{\bratable}[2]{\mymath{#1\brasymb \left\{\begin{array}{ll}#2\end{array}\right\}}}

\newcommand{\MM}[1][]{\mymath{M_{#1}}}
\newcommand{\MMd}[1][]{\mymath{M'_{#1}}}

\newcommand{\proc}[2]{\mymath{#1 \blacktriangleright  \left[#2\right]}}
\newcommand{\store}[3]{\mymath{#1\blacktriangleright [\pdata{#2}{#3}]}}

\newcommand{\Ga}[1][]{\mymath{\mathsf{\Gamma}_{#1}}}

\newcommand{\proves}{\mymath{\mathbin{\vdash}}}
\newcommand{\as}{\mymath{\mathbin{:}}}

\newcommand{\proj}[2]{\mymath{#1\rceil#2}}

\newcommand{\mrg}{\ensuremath{\sqcap}}

\newcommand{\subt}{\ensuremath{\leqslant}}

\newcommand{\inputvar}[1]{\ensuremath{(#1)}}

\newcommand{\outputvar}[1]{\ensuremath{\langle#1\rangle}}

\newcommand{\inpprefix}[2]{\ensuremath{#1?\inputvar{#2}}}

\newcommand{\outprefix}[2]{\ensuremath{#1!\outputvar{#2}}}

\newcommand{\psep}{\ensuremath{.\,}}
\newcommand{\If}{\ensuremath{\mathtt{if}}\xspace}
\newcommand{\Then}{\ensuremath{\mathtt{then}}\xspace}
\newcommand{\Else}{\ensuremath{\mathtt{else}}\xspace}

\newcommand{\match}[2]{\ensuremath{#1 = #2}}

\newcommand{\ifelse}[4]{\ensuremath{\If\ \match{#1}{#2}\ \Then\ #3\ \Else\ #4}}

\newcommand{\inact}{\mymath{\mathbf{0}}}
\newcommand{\Par}{\mymath{\mathrel{|}}}

\newcommand{\ceq}[2]{%
  \par\vspace{#1pt}
  \centerline{
  $#2$}
}

\newcommand{\privatet}[1]{\ensuremath{\mathsf[#1]}\xspace}
\newcommand{\hit}{\ensuremath{\_}}
\newcommand{\pdtype}[2]{\ensuremath{#1\otimes#2}}
\newcommand{\antype}[3]{\ensuremath{#1{[\pdtype{#2}{#3}]}}}

\newcommand{\Vars}{\mymath{\mathcal{V}}}
\newcommand{\Constants}{\mymath{\mathcal{C}}}

\newcommand{\Ids}{\mymath{\mathsf{Ids}}}
\newcommand{\id}[1][]{\mymath{\mathsf{id_{#1}}}}

\newcommand{\hid}{\mymath{\_}}

\newcommand{\pdata}[2]{\mymath{#1\otimes#2}}

\newcommand{\ii}{\ensuremath{\iota}\xspace}

\newcommand{\scong}{\mymath{\equiv}}

\newcommand{\subst}[2]{\mymath{\{^{#1}/_{#2}\}}}

\newcommand{\trans}[1]{\mymath{\stackrel{#1}{\longrightarrow}}}
\newcommand{\mtrans}[1]{\mymath{\stackrel{#1}{\longrightarrow^*}}}

\newcommand{\actwithdrawid}[2]{\mymath{#1\triangleleft\withdrawsymb@#2}}
\newcommand{\actstorewithdraw}[2]{\mymath{\overline{#1}\triangleright\withdrawsymb@#2}}

\newcommand{\g}[1][]{\mymath{\mathsf{g}_{#1}}}
\newcommand{\ti}{\mymath{\mathsf{i}}}
\newcommand{\tid}{\mymath{\mathsf{t_i}}}

\newcommand{\descr}[1]{(\textit{\footnotesize #1})}

\newcommand{\es}{\ensuremath{\emptyset}}

\newcommand{\cat}{,}

\newcommand{\trulefont}[1]{\ensuremath{[\text{\footnotesize$\mathsf{T#1}$}]}\xspace}

\newcommand{\TPdata}{\trulefont{Pd}}
\newcommand{\TVdata}{\trulefont{Vd}}
\newcommand{\THdata}{\trulefont{Hd}}
\newcommand{\TVar}{\trulefont{Var}}
\newcommand{\TVarRec}{\trulefont{VRec}}
\newcommand{\TPlace}{\trulefont{Plc}}

\newcommand{\TData}{\trulefont{Data}}

\newcommand{\TInact}{\trulefont{Inact}}
\newcommand{\TOut}{\trulefont{Out}}
\newcommand{\TOutPD}{\trulefont{OutPD}}

\newcommand{\TInp}{\trulefont{Inp}}
\newcommand{\TSel}{\trulefont{Sel}}
\newcommand{\TBranch}{\trulefont{Branch}}
\newcommand{\TInpPD}{\trulefont{InpPD}}

\newcommand{\TRec}{\trulefont{Rec}}

\newcommand{\TIf}{\trulefont{If}}

\newcommand{\TNet}{\trulefont{Net}}

\newcommand{\TSub}{\trulefont{Sub}}

\newcommand{\lts}[1]{\ensuremath{[\text{\footnotesize$\mathsf{#1}$}]}\xspace}

\newcommand{\outpd}[2]{\outprefix{#1}{#2} \psep}
\newcommand{\inppd}[2]{\inpprefix{#1}{#2} \psep}

\newcommand{\LSOut}{\lts{SOut}}

\newcommand{\LTrA}{\lts{True}}
\newcommand{\LFlA}{\lts{False}}

\newcommand{\Alpha}{\lts{Struct}}

\newcommand{\LSInp}{\lts{SInp}}

\newcommand{\Labels}{\mymath{\mathcal{L}}}
\colorlet{mypurple}{red!40!blue}
\definecolor{myorange}{HTML}{D55E00}

\newcommand{\lbl}[1]{\mymath{\textcolor{myorange}{\mathsf{#1}}}}
\newcommand{\lab}[1][]{\lbl{\ell_{#1}}}

\colorlet{myred}{red!70!black}
\newcommand{\type}[1]{\ensuremath{\textcolor{myred}{\mathtt{#1}}}\xspace}

\newcommand{\integer}{\type{int}}
\newcommand{\bool}{\type{bool}}

\colorlet{myblue}{blue!60!black}

\newcommand{\tstore}[1]{\ensuremath{\textcolor{myblue}{\mathsf{#1}}}\xspace}

\newcommand{\annota}[1][]{\tstore{\alpha_{#1}}}
\newcommand{\annotb}[1][]{\tstore{\beta_{#1}}}
\newcommand{\annotg}[1][]{\tstore{\gamma_{#1}}}

\colorlet{myamber}{orange!80!yellow}

\colorlet{mygreen}{green!45!black}

\newcommand{\pcolour}[1]{\textcolor{mygreen}{#1}}
\newcommand{\pt}[1]{\ensuremath{\pcolour{\mathsf{#1}}}\xspace}
\newcommand{\p}[1][]{\pt{p_{#1}}}
\newcommand{\q}[1][]{\pt{q_{#1}}}
\newcommand{\rr}[1][]{\pt{r_{#1}}}
\newcommand{\pd}[1][]{\pt{p'_{#1}}}

\newcommand{\patient}[1][]
{\pt{patient_{#1}}}

\newcommand{\generalpract}[1][]
{\pt{nurse{#1}}}

\newcommand{\laboratory}[1][]{\pt{lab_{#1}}}
\newcommand{\diagnosiseng}[1][]
{\pt{practitioner_{#1}}}
\newcommand{\infostore}[1][]{\tstore{info_{#1}}}
\newcommand{\symptomsstore}[1][]
{\tstore{symptoms_{#1}}}
\newcommand{\lorderstore}[1][]{\tstore{lorder_{#1}}}
\newcommand{\resultsstore}[1][]
{\tstore{results_{#1}}}
\newcommand{\diagnstore}[1][]
{\tstore{diagnosis_{#1}}}

\newcommand{\basicinfo}[1][]{\type{basicInfo{#1}}}
\newcommand{\symptoms}[1][]{\type{sympt{#1}}}
\newcommand{\laborder}[1][]{\type{labOrder{#1}}}
\newcommand{\labresults}[1][]{\type{labRes{#1}}}
\newcommand{\diagnosticreport}[1][]{\type{diagnReport{#1}}}

\newcommand{\labyes}{\lbl{lab}}

\newcommand{\labno}{\lbl{noLab}}

\newcommand{\U}[1][]{\mymath{\mathsf{U_{#1}}}}
\newcommand{\Ud}[1][]{\mymath{\mathsf{U'_{#1}}}}

\newcommand{\G}[1][]{\ensuremath{\mathsf{G_{#1}}}\xspace}

\newcommand{\pass}[3]{#1 \to #2{:} \langle#3\rangle.\,}

\newcommand{\choice}[4]{ \mymath{#1 \to #2: \left\{{#3}\right\}_{#4}} }

\newcommand{\choicetable}[3]{
    \ensuremath{
    #1 \to #2:
    \left\{ 
        \begin{array}{ll}
            #3
        \end{array}
    \right\} 
    }
}

\newcommand{\ginact}{\mathsf{end}}

\newcommand{\pts}[1]{\ensuremath{\mathsf{pts}(#1)}}

\newcommand{\gcons}[4]{
    #1 \setminus #2 \stackrel{#4}{\rightarrow} #3
}

\newcommand{\gconsdef}[3]{
    #1 \stackrel{#3}{\rightarrow} #2
}

\newcommand{\gconsShort}[2]{
    #1 \setminus #2
}

\newcommand{\act}{\ensuremath{\mathsf{act}}}

\newcommand{\local}[1][]{\mymath{T_{#1}}}
\newcommand{\locald}[1][]{\mymath{T'_{#1}}}

\newcommand{\comment}[1]{}
\newcommand{\tout}[2]{#1\outsymb#2\sep}
\newcommand{\tinp}[2]{#1\inpsymb#2\sep}
\newcommand{\tselect}[2]{#1\oplus\{#2\}}
\newcommand{\tselecti}[3]{#1\oplus\{#2\}_{#3}}
\newcommand{\tbranch}[2]{#1\&\{#2\}}
\newcommand{\tbranchi}[3]{#1\&\{#2\}_{#3}}
\newcommand{\tinact}{\mathsf{end}}

\newcommand{\tchoicetable}[2]{
    \ensuremath{
    #1 \;\&
    \left \{
        \begin{array}{ll}
            #2
        \end{array}
    \right\}
    }
}

\begin{document}
\maketitle

\begin{abstract}
The General Data Protection Regulation (GDPR) establishes purpose
limitation as a fundamental constraint on personal data processing: 
personal data must be collected, stored, and processed strictly in
accordance with explicitly specified purposes. 
Therefore,
systems are required not only to declare the purposes under which
personal data are processed, but also to ensure that their
runtime behaviour 
remains aligned with the declared purposes.
Yet, in mainstream software engineering practice, purposes are often
treated as informal declarations, largely disconnected from system
behaviour and, therefore, not amenable to rigorous reasoning
about purpose compliance. This gap becomes particularly problematic
in distributed systems, where personal data may flow across multiple
entities and evolve through complex communication patterns.
To address this challenge,  recent works propose
a more elaborate treatment of purposes based on structured,
action-oriented representations of the data-processing interactions
involved in their fulfilment. 
Building on these insights, we introduce a formal, purpose-aware
framework grounded in multiparty session types
in which purposes are modelled as structured interaction protocols
among system entities. 
Within our framework, system implementations are specified using a process calculus that captures the semantics of distributed interactions and features private data as a first-class entity.
Furthermore, we define a type system that verifies compliance between declared purposes and system models, and we establish subject reduction and purpose fidelity results, thereby ensuring that well-typed systems do not deviate from their specified purposes during execution. 
We demonstrate our approach through a case study involving a healthcare system. 
Ultimately, our objective is to evolve this formal framework into a software-engineering-oriented approach that unifies purpose modelling and compliance verification within a lifecycle-driven methodology, thus
enabling 
a practically applicable privacy-by-design process.
\end{abstract}

\section{Introduction}
\label{introduction}
The increasing use of software systems that store and process personal data 
raises major concerns about data privacy. In response, governments and legal institutions have enacted regulations such as the European Union General Data Protection Regulation (GDPR)~\cite{eugdpr}. A key concept in the GDPR and other privacy directives is \textit{purpose limitation}, which requires that personal data be collected for specified and explicit purposes, and not further processed in a manner incompatible with those purposes.
Therefore, systems are required to specify the processing purposes they may pursue when handling personal data. %
For example, in healthcare, a doctor may process a patient's medical data for the purpose of performing a {\em diagnosis}, while the accounting department may use the patient's address for the purpose of {\em issuing an invoice}. In contrast, repurposing the same medical records for an unrelated objective, such as {\em marketing}, would constitute a purpose incompatible with the original intent of collection and would thus be disallowed under the principle of purpose limitation.

While purposes can be specified at the policy level, ensuring that the actual behaviour of a system remains aligned with the declared purposes is considerably more complex. 
Yet, in mainstream software engineering practice, 
existing approaches 
often treat purposes as informal declarations, e.g., as simple textual labels~\cite{tokas2020formal,alshareef2022precise}, disconnected from system behaviour and lacking formal guarantees. To address this challenge, the need of providing formal 
semantics to the notion of purpose has been strongly
advocated and, in the formal landscape, several works have adopted
such semantics based on structured,
action-oriented representations of the data-processing interactions
involved in a purpose fulfilment~\cite{RiahiKG17,tschantz2011semantics,TschantzDW12,jafari2014framework}, 
Building on this perspective and informed by GDPR analysis~\cite{vanezi2019gdpr}, in our recent work
we have defined purposes as \textit{structured interaction protocols over data}~\cite{EVpurpose,vanezi2025privacy,vanezi2020dialogop}.
Under this approach purposes are captured as scenarios of 
communication actions and interactions between system entities, 
that specify the patterns of  exchange, storage, reading, and writing of personal data. 
Practical applicability in software engineering is demonstrated through purpose-aware UML sequence diagrams~\cite{EVpurpose,vanezi2025privacy}, which constitute our proposed modelling notation for representing and reasoning about purposes at the design level (see Fig. 1 and relevant discussion below).

%


The current work provides a formal foundation for this definition by modelling purposes as global types within a purpose-aware multiparty session type framework. System implementations are specified using a process calculus that captures the semantics of distributed interactions, and features private data as a first-class entity.
Unlike standard uses of multiparty session types, where global types specify only communication correctness, here they additionally encode purpose constraints: they restrict not only the flow of messages, but also the flow, access, and usage of personal data, including interactions with data stores.
We then define a type system that verifies compliance between declared purposes and system models and establish subject reduction and purpose fidelity results, thereby ensuring that well-typed systems do not deviate from their specified purposes during execution. 

%
\vspace{-0.4cm}
\paragraph{\textbf{An example of purpose.}}
A representative example of purpose is a diagnostic process taking place in a medical system, where personal data are exchanged between multiple parties. In this purpose, a patient provides their data (e.g., basic information and symptoms), and requests a diagnosis. These data should only be processed and accessed according to the purposes specified in the system's privacy policy, to initiate, execute, and fulfil the requested process. 
Such a process might involve additional participants, such as a {\em nurse} or a {\em general practitioner}, and operations such as {\em reading symptoms} or {\em ordering a laboratory test}, while different interactions might occur depending on intermediate decisions. In any case, every action and interaction must follow the specified purpose protocol. 
Furthermore, while some entities may need to process the data to perform their task, others may only relay or store information without accessing its contents.

\begin{figure} [h!]
    \centering
    \includegraphics[width=0.8\linewidth]{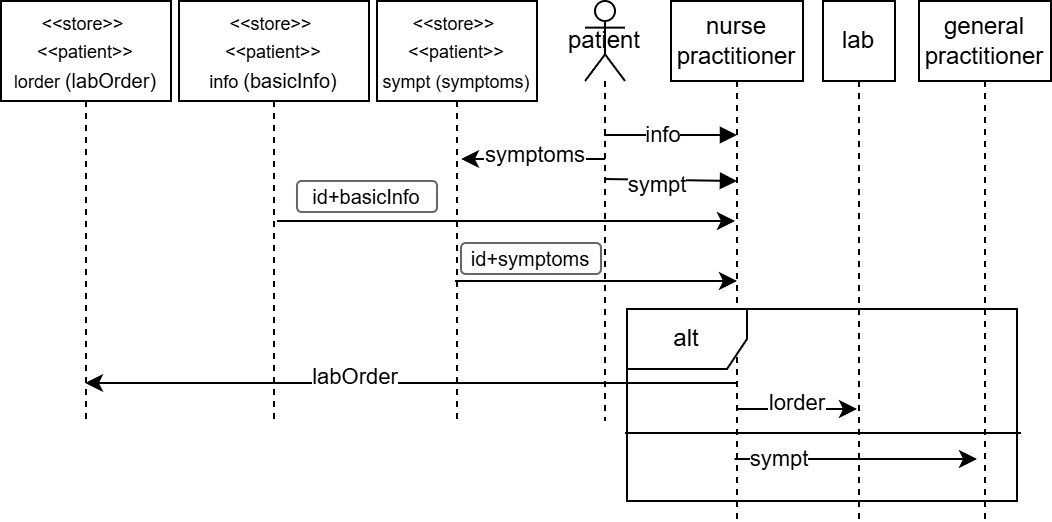}
    \caption{Part of the diagnostic purpose illustrated via a purpose-aware sequence diagram~\cite{vanezi2025privacy}}
    \label{fig:diagnosispurpose}
\end{figure}
Figure~\ref{fig:diagnosispurpose} presents part of such a diagnostic purpose via a purpose-aware sequence diagram~\cite{EVpurpose,vanezi2025privacy}. 
Participants, represented by lifelines, include system entities and additionally personal data repositories called stores (specified through the stereotype <<\textit{store}>>). Such repositories belong to a specific actor (specified by the stereotype <<\textit{owner}>>), and hold data of a specific type defined in brackets, e.g., (basicInfo). Store names (e.g., info), called references, are used to interact with them. Exchanges with such repositories signify accessing personal data. 
Such a diagram illustrates how personal data are intended to be exchanged between the participants, and how access to data is dictated by the defined purpose, including branching interactions which determine how additional entities are involved.


The purpose illustrated above specifies that a patient first communicates their basic information via the corresponding reference (info) to a nurse practitioner, followed by writing their symptoms in the corresponding personal data repository (sympt), and forwards them to the nurse. In turn, the nurse practitioner reads both pieces of personal data from the repositories, and then decides whether a laboratory exam is needed, in which case the nurse prescribes the order in the corresponding repository (labOrder) and forwards its reference to the lab; or, if a laboratory exam is not needed, the nurse proceeds with sending the symptoms repository reference to the general practitioner to perform a diagnosis. 
%
Our aim in this paper is to formalize the approach through multiparty session types and evolve it into a framework that
enables providing formal guarantees of purpose compliance. The complete case study is detailed in Section~\ref{casestudy}.

%
%
%
%
\paragraph{\textbf{Related Work.}}
\label{background} 
Existing approaches for reasoning about purposes in software engineering such as~\cite{cambronero2024towards,pedroza2021model,ye2023mbipv} describe purposes using informal or semi-formal specification with limited semantic grounding, while others represent and handle purposes through textual labels~\cite{tokas2020formal,alshareef2022precise}. Such approaches do not provide
a basis for precise analysis of whether system behaviour conforms to purpose-based specification. 

To address this need, several formal approaches have also been proposed, approaching purposes as workflows. 
The work in~\cite{RiahiKG17} models purposes as Petri net workflows and verifies compliance against actor models using model checking. 
Markov Decision Processes are employed in ~\cite{tschantz2011semantics,TschantzDW12}, while a planning-based formalism is proposed to audit systems against privacy policies. A similar  
approach is adopted in \cite{basin2018purpose}, where purposes are identified with business processes, and formal models of inter-process communication are used to derive or audit privacy policies.
In~\cite{jafari2011towards}, the authors introduce a semantic model for purpose-based privacy policies, together with a modal logic and the corresponding model-checking procedure for compliance verification. 
Other works like~\cite{MasellisGR15}
define temporal logic-based semantics and run-time monitoring methodologies for the enforcement of purpose-based privacy policies.

Complementing and extending these approaches, our work models purposes as structured communication protocols between interacting entities. It captures personal data flows and access patterns, and enables static verification of purpose compliance at design time through a type system grounded in multiparty session types. 
Additionally, the practical application of our approach in
software engineering practice is direct, through purpose-aware UML sequence diagrams~\cite{EVpurpose,vanezi2025privacy}, enabling future integration into a lifecycle-driven methodology.

\paragraph{Contribution.}
This work presents the following contributions:
\begin{enumerate}
\item \textbf{Formal modelling of systems} using a process calculus that explicitly captures personal data processing (Sect.~\ref{calculus});
\item \textbf{Formal specification of processing purposes} through a purpose-aware type language (Sect. ~\ref{purpose});
\item \textbf{Formal verification of system compliance with declared purposes} via a type system (Sect.~\ref{typesystem}), supported by proofs for subject reduction and purpose fidelity (Sect.~\ref{theorems});
\item \textbf{A realistic case study}, demonstrating compliant and non-compliant behaviours captured by our framework (Sect.~\ref{casestudy}).
\end{enumerate}
Section~\ref{concl} discusses conclusions and future work. Detailed proofs, auxiliary lemmas, and typing of the case study, are presented in
~\cite{vanezigdprmst}. 

\section{The Calculus}
\label{calculus}
In this section we present a process calculus for modelling systems that process personal data. The calculus integrates (i) the notions of personal data and personal data stores~\cite{KP-LMCS17}, to capture personal data processing operations, and
(ii) multiparty session types (MPST)~\cite{yoshida2019very}, to structure interactions as communication protocols among participants.
%
MPSTs provide a typed foundation for specifying structured communication among distributed entities, ensuring properties such as communication safety and progress. 
%
%
The goal of the calculus is to provide a formal model in which both communication 
and personal data manipulation can be explicitly represented and analysed. In 
particular, it enables us to model systems as workflows of message-passing actions and data-processing operations, forming the basis for reasoning about purpose compliance.
In contrast to standard MPST-based calculi, our model explicitly incorporates personal data stores and controlled data access via references.

Based on~\cite{KP-LMCS17} and following the GDPR, we assume that personal data are
data associated to individuals. As such, we model personal data as structures that associate 
constants, that is, pieces of information, with identifying pieces of information, 
which we simply refer to as \emph{identities}. Such personal data used in an information system are typically  
stored within a database of the system. We capture
such databases as a collection of {\em stores}, where a store is
encoded as a high-level process term associated with operations for the manipulation
of personal data. This manipulation takes place with the use of a special kind
of names called {\em references}. 
In the sequel, we describe these concepts within a multi-party, synchronous, session types framework. For simplicity, we eliminate shared channels for session initiations. 
However, the framework can be easily extended to accommodate multiple sessions.

\subsection{Syntax}
Figure~\ref{fig:syntax} defines the syntax of the proposed calculus. The calculus assumes the following basic structures:
i) the set of {\em session participants} $\p, \q, \ldots$;
ii) the set of variables $\Vars$, ranged over by $x, y, z, \ldots$;
iii) the set of store references $\mathcal{R}$, ranged over by $\rr, \rr',\ldots$;
%
%
iv) the set of constants $\Constants$, ranged over by $\const, \ldots$;
v) the set of identities $\Ids$, ranged over by $\id, \id[1], \id[2], \ldots$, with the {\em anonymous} identity \hid\ supporting personal data anonymisation; and
vi) the set of labels $\Labels$, ranged over by $\lab, \lab[1], \lab[2], \dots$.

\begin{figure}[h]
    \[
    \begin{array}{rcll}
    \hline \\
        v       &\bnfis& \const \bnfbar \rr & \descr{values}
        \\[1mm]
        t		&\bnfis& v \bnfbar x & \descr{terms}
        \\[1mm]
        u       &\bnfis& \rr \bnfbar x & \descr{store terms}
        \\[1mm]
        \ii   &\bnfis& \id \bnfbar \hid & \descr{identity}
        \\[1mm] 
        k		&\bnfis& \pdata{x}{y} \bnfbar \pdata{\hid}{x} & \descr{placeholders}
        \\[1mm]
        d		&\bnfis&  \pdata{\ii}{\const}\bnfbar k & \descr{data terms}
        \\[4mm]
        \PP   &\bnfis&    \outp{\p}{t} \PP \bnfbar \inpp{\p}{x} \PP \bnfbar\selp{\p}{\lab} \PP \bnfbar \brapi{\p}{\lab[i]: \PP[i]}{i \in I}  
        &\descr{processes}
        \\[1mm]
        &\bnfbar&  \outp{u}{d} \PP \bnfbar \inpp{u}{k} \PP  \bnfbar	\ifelse{a}{b}{P}{P}\;\;\;\;\;
        \\[1mm]
        &\bnfbar& X \bnfbar  
        \mu X. P \bnfbar \inact
        \\[4mm] 
        \MM   &\bnfis&    \proc{\p}{\PP} \bnfbar \store{\rr}{\id}{\const} \bnfbar \MM \Par \MM
        &  \descr{networks}
        \\ \\
        \hline
    \end{array}
    \]
	\caption{Syntax of the calculus}
	\label{fig:syntax}
\end{figure}

Values include constants $\const$ and references $\rr$. Terms $t$ include values and variables $x$, whereas store terms $u$ include references and variables.
Metavariable \ii\ ranges over identities, \id, and the anonymised identity \hid. The symbol $k$ is a variable placeholder for personal data, ranged over by \pdata{x}{y} and \pdata{\hid}{x}. The symbol $v$ ranges over personal data \pdata{\ii}{\const}, associating data $\const$ with an identity \ii, signifying the owner of the data, and personal data variable $k$. 
Personal data of the form \pdata{\hid}{\const}, denote anonymised personal data.

Processes implement the behaviour of participants. 
Process $\outp{\p}{t} \PP$ sends term $t$ to participant $\p$ and continues as \PP, whereas process $\inpp{\p}{x} \PP$ denotes receiving a value from participant $\p$, substituting it into variable $x$, and continuing as \PP.
The \emph{select process} $\selp{\p}{\lab} \PP$ selects a label \lab on participant \p and  continues as \PP. Dually, the branch process \brapi{\p}{\lab[i]: \PP[i]}{i \in I} waits for the selection of one of the labels \lab[i] from participant \p and then proceeds with the corresponding process $\PP[i]$, $i \in I$.
Personal data processing uses the special store references ($\rr$), which provide access to processes that store personal data. Processes do not exchange personal data directly; instead, they exchange references. Any process with access to a store reference can use it to write data to or read data from the corresponding store.
Process $\outp{u}{d} \PP$ represents storing personal data $d$ on reference $u$ and  continuing as \PP. Dually, process $\inpp{u}{k} \PP$ denotes reading personal data via reference $u$ and substituting the retrieved data within \PP before continuing. 

The matching construct \ifelse{a}{b}{P}{Q} evolves as $P$ if the condition $\match{a}{b}$ holds, and as $Q$ otherwise, where $a$ and $b$ range over values (e.g., constants). 
A conditional is used to decide internal process control. When combined with subtyping, a conditional is the mechanism for performing internal choice, i.e., allowing multiple select labels. 
Process $X$ ranges over recursion variables, and process $\mu X.P$ is a recursive process. Finally, \inact represents an inactive process.
A {\em role} network, \proc{\p}{\PP}, associates a participant \p with a process \PP.
The notation \store{\rr}{\id}{{\const}} represents a \textit{personal data store}, where personal data are accessed through the references $\rr$.
%
The syntax supports the parallel composition of networks, $\MM[1] \Par \MM[2]$. 
\begin{example}[Simplified Medical Workflow Network]\label{ex1}
    \rm
    We illustrate the use of the calculus through a simplified fragment of the 
 diagnostic workflow introduced in Section~\ref{introduction} and developed fully in Section~\ref{casestudy}.
    In network
    $M_1$ below, a \patient sends a reference ($\rr[s]$) to a data store containing their symptoms (\pdata{\id}{\symptoms_1}) to a \generalpract, who then accesses the personal data from the store and proceeds with further actions. 
    %
\ceq{4}{
        \MM[1]  =             \proc{ \patient }{ \outp{\generalpract}{\rr[s]} P }
                \Par    \proc{ \generalpract }{ \inpp{\patient}{z} \inpp{z}{y \otimes x_{sympt}} 
                Q}
                \Par    \store{\rr[s]}{\id}{\symptoms}
} 
    A variation of the above example ensures that the personal data read by the  \generalpract is anonymised:
\ceq{4}{
        \MM[2]  =             \proc{ \patient }{ \outp{\generalpract}{\rr[s]} P }
                \Par    \proc{ \generalpract }{ \inpp{\patient}{z} \inpp{z}{\hid \otimes x_{sympt}}
                Q }
                \Par    \store{\rr[s]}{\id}{\symptoms}
         \label{examplehid}
    }
    The next network demonstrates a choice implemented via select and branch processes:
\ceq{4} {
         \MM[3]  =    \proc{\generalpract}{\selp{\diagnosiseng}{\labyes} \inact} \Par \proc{\diagnosiseng}{\brapi{\generalpract}{\labyes: \PP[1], \labno: \PP[2]}{}}   
         \label{examplebrsel}
    }
    Network \MM[3] demonstrates the case in which the \diagnosiseng is expecting from the \generalpract either the decision that laboratory exam is needed (label \labyes) and proceed as \PP[1], or the decision that no exam will take place (label \labno), in which case it will proceed as \PP[2]. At the same time the \generalpract makes the choice that a lab exam is needed, by sending label \labyes to the \diagnosiseng.
    \qed
\end{example}

\subsection{Operational Semantics}
The semantics of the process calculus define substitution in two cases:
(i) \(\PP\subst{v}{x}\), which substitutes all free occurrences of \(x\) in
\(\PP\) with \(v\) ; and
(ii) \(\PP\subst{\pdata{\ii}{\const}}{k}\), which substitutes all free
occurrences of \(k\) in \(\PP\) with \(\pdata{\ii}{\const}\).
The definitions can be found in
~\cite{vanezigdprmst}.
In particular, we note that substitution on private-data placeholders accounts for
anonymous substitution:
\[
\setlength{\arraycolsep}{2pt}
\begin{array}{rclcrcl}
    \pdata{x}{y}\subst{\pdata{\ii}{\const}}{\pdata{x}{y}}
    &=&
    \pdata{\ii}{\const}
    &&
    \pdata{\hid}{y}\subst{\pdata{\ii}{\const}}{\pdata{\hid}{y}}
    &=&
    \pdata{\hid}{\const}
    \\[1mm]
    \pdata{z}{w}\subst{\pdata{\ii}{\const}}{\pdata{x}{y}}
    &=&
    \pdata{z}{w}
    &\text{if}&
    z \neq x \text{ or } w \neq y .
\end{array}
\]

Figure~\ref{fig:lts} presents the operational semantics of the process calculus, defined in terms of a structural congruence relation $\equiv$.
Structural congruence is the least relation generated by the rules in the first two lines of Figure~\ref{fig:lts}. It establishes the associativity and commutativity of the parallel composition operator and treats a terminated role of the form \proc{\p}{\inact} as a neutral element.

\begin{figure} [h]
    \[
    \begin{array}{c}
    \hline \\
        \begin{array}{c}
            \MM \Par \proc{\p}{\inact} \equiv \MM
            \qquad
            \MM[1] \Par \MM[2] \equiv \MM[2] \Par \MM[1]        
            \\[3mm]
            (\MM[1] \Par \MM[2]) \Par \MM[3] \equiv \MM[1] \Par (\MM[2] \Par \MM[3])
            \qquad 
            \proc{\p}{\mu X.\PP} \Par \MM \equiv \proc{\p}{\PP \subst{\mu X.P}{X}} \Par \MM
        \end{array}
        \\[10mm]
        \begin{array}{rrcl}
            \LComm &
                        \proc{\p}{\outp{\q}{v} \PP[1]} \Par \proc{\q}{\inpp{\p}{x} \PP[2]} \Par \MM &\trans{}& \proc{\p}{\PP[1]} \Par \proc{\q}{\PP[2] \subst{v}{x}} \Par \MM
            \\[3mm]
            \LBranch &
                        \proc{\p}{\selp{\q}{\lab[j]}\PP} \Par \proc{\q}{\brapi{\p}{\lab[i]: \PP[i]}{i\in I}} \Par \MM &\trans{}& \proc{\p}{\PP} \Par \proc{\q}{\PP[j]} \Par \MM \;\; j\in I
            \\[3mm]
            \LSOut &
                        \proc{\p}{ \outp{\rr}{\pdata{\ii}{\const}} \PP} \Par \store{\rr}{\id}{\constd} \Par \MM &\trans{}& \proc{\p}{\PP} \Par \store{\rr}{\id}{\const} \Par \MM
            \\[3mm]
            \LSInp &
                        \proc{\p}{ \inpp{\rr}{k} \PP}  \Par \store{\rr}{\id}{\const} \Par \MM &\trans{}& \proc{\p}{\PP \subst{\pdata{\id}{\const}}{k}} \Par \store{\rr}{\id}{\const} \Par \MM
            \\[3mm]
        
            \LTrA &
               \proc{\p}{\ifelse{a}{a}{P}{Q}} \Par \MM
    		      &\trans{}&
                \proc{\p}{P} \Par \MM
            \\[3mm]
            \LFlA &
                \proc{\p}{\ifelse{a}{b}{P}{Q}}  \Par \MM
                &\trans{}&
                \proc{\p}{Q}  \Par \MM \qquad a \neq b
            \\[3mm]
            \Alpha
            & \multicolumn{3}{c}{
                \tree {
                    \MM[1] \equiv \MMd[1] \quad \MMd[1] \trans{} \MMd[2] \quad  \MMd[2]\equiv \MM[2]
                }{
                     \MM[1] \trans{} \MM[2]
                }
            }
        \end{array}
        \\ \\    
        \hline
    \end{array}
    \]
	\caption{Structural Congruence and Operational Semantics}
	\label{fig:lts}
\end{figure}

Rule \LComm describes message communication between two parallel roles. Role \p sends a value $v$ to role \q, which receives it in variable $x$. 
After the interaction, the former proceeds as $\proc{\p}{P_1}$, and the latter as $\proc{\q}{P_2 \subst{v}{x}}$, where all free occurrences of $x$ in $P_2$ are substituted by $x$.
Rule \LBranch describes synchronisation between two parallel roles to execute a choice interaction. Role \p sends a label $\lab[j]$ to role \q, while role \q receives the label from role \p. After the interaction, the former proceeds as $\proc{\p}{P}$, and the latter as $\proc{\q}{P_j}$, following label \lab[j].
Rule \LSOut describes a role storing a value \pdata{\id}{\const} in the store corresponding to reference $\rr$. After the storing action, the store updates its data value, while the role proceeds as $\proc{\p}{P}$.
Rule \LSInp describes a role reading a value \pdata{\id}{\const} into variable $k$ from the store corresponding to reference $\rr$. After the read action, the role proceeds as $\proc{\p}{\subst{\pdata{\id}{\const}}{k}}$, with $k$ substituted by \pdata{\id}{\const}.
%
Finally, rule \Alpha defines reduction under structural equivalence. 
We write \mtrans{} for the reflexive and associative closure of relation \trans{}.
\begin{example}[Simplified Medical Workflow Network Transition] 
\rm
We demonstrate the application of the transition rules \LComm and \LSInp, to networks $\MM[1]$ and $\MM[2]$ from Example~\ref{ex1}, as follows:
\begin{eqnarray*}
        \MM[1] & = &            \proc{ \patient }{ \outp{\generalpract}{\rr[s]} P }
        \Par    \proc{ \generalpract }{ \inpp{\patient}{z} \inpp{z}{y \otimes x_{sympt}} 
        Q } 
        \Par    \store{\rr[s]}{\id}{\symptoms}\\
        &\trans{}&
        \proc{ \patient }{ P }
        \Par    \proc{ \generalpract }{  \inpp{\rr[s]}{y \otimes x_{sympt}} 
        Q } 
        \Par    \store{\rr[s]}{\id}{\symptoms}\\
        &\trans{}&
        \proc{ \patient }{ P }
        \Par    \proc{ \generalpract}
        {  
        Q\subst{\pdata{\id}{\symptoms}}{y \otimes x_{sympt}} 
        }
        \Par    \store{\rr[s]}{\id}{\symptoms}
\\
\\
         \MM[2] & = &            \proc{ \patient }{ \outp{\generalpract}{\rr[s]} P }
        \Par    \proc{ \generalpract }{ \inpp{\patient}{z} \inpp{z}{\hid \otimes x_{sympt}} 
        Q } 
        \Par    \store{\rr[s]}{\id}{\symptoms}\\
        &\trans{}&
    \proc{ \patient }{ P }
        \Par    \proc{ \generalpract }{  \inpp{\rr[s]}{\hid \otimes x_{sympt}} 
        Q } 
        \Par    \store{\rr[s]}{\id}{\symptoms}\\
        & \trans{}&
        \proc{ \patient }{ P }
        \Par    \proc{ \generalpract}
        {  
        Q\subst{\pdata{\id}{\symptoms}}{\hid \otimes x_{sympt}}
        }
        \Par    \store{\rr[s]}{\id}{\symptoms}
         \label{example} 
    \end{eqnarray*} 
\end{example}
\vspace{-3pt}
Note that the substitution function ensures that in 
 the continuation process of the \generalpract
in \MM[2] the patient's data are read anonymized.
\section{Defining Purposes via a Type Language}
\label{purpose}
We now proceed to propose a formal type language to rigorously model purposes.
The language extends multiparty session types, enabling the definition of a purpose by formalising how personal data flows among system entities in order to achieve the purpose.
In particular, we represent a purpose by means of a {\em global type}, which specifies the interaction protocol that participants must follow to achieve the intended objective.
Subsequently, {\em local types} serve to characterise the role-specific realisation of the purpose. The relationship between global and local types is established through projection.

\subsection{Global and Local Types} 
\begin{figure}[h]
\[
    \begin{array}{rccll}
\hline
        \\
        & && \quad \mathsf{\annota, \annotb, \annotg, \dots} & \descr{annotations}
        \\[2mm]
        & \quad \g &\bnfis& \integer \bnfbar \bool \bnfbar \dots & \descr{Ground Types} 
        \\[2mm]
        & \quad \tid &\bnfis& \ti \bnfbar \hit   & \descr{Id Types} 
        \\[2mm]
        & \quad \mathsf{U}&\bnfis& \g  
        \bnfbar \antype{\annota}{\tid}{\g}  
        & \descr{Exchange Types}
        \\[2mm]
        & \quad  \G &\bnfis& \pass{\p}{\q}{\U} \G
         &\descr{Value Exchange}
        \\[1mm]
        &   &\bnfbar&  \pass{\annota}{\p}{\pdtype{\tid}{\g}} \G
        &\descr{Personal Data Read} \\[1mm]
        &   &\bnfbar&  \pass{\p}{ \annota }{\pdtype{\tid}{\g}}  \G
        &\descr{Personal Data Storage} \\[1mm]
        & &\bnfbar&  \choice{\p}{\q}{\lab[i] : \G_i }{i \in I} 
        \qquad \qquad &\descr{Select/Branch}
        \\[1mm]
        & &\bnfbar&  t  \bnfbar \mu t.\G
        \qquad \qquad &\descr{Recursion}
     \\[1mm]
        &   &\bnfbar&  \ginact & \descr{Inact}
        \\[3mm]
        & T &\bnfis& \tinp{\p}{\U} \local & \descr{Value Input}
        \\[1mm]
        && \bnfbar & \tout{\p}{\U} \local &\descr{Value Output}
        \\[1mm]
        && \bnfbar & 
        \tinp{\annota}{\pdtype{\tid}{\g}} 
        \local & \descr{Personal Data Value Input} 
        \\[1mm]
        && \bnfbar & 
        \tout{\annota}{\pdtype{\tid}{\g}}
        \local 
        & \descr{Personal Data Value Output}
        \\[1mm]
        && \bnfbar & \tselecti{\p}{\lab[i]: \local[i]}{i \in I}  & \descr{Selection}
        \\[1mm]
        &&\bnfbar& \tbranchi{\p}{\lab[i]: \local[i]}{i \in I}  & \descr{Branching}
        \\[1mm]
        & &\bnfbar&  t  \bnfbar \mu t.\local
        \qquad \qquad &\descr{Recursion}
        \\[1mm]
        && \bnfbar & \tinact \qquad & \descr{Inact} 
        \\
        \\
        \hline
    \end{array}
\]
\caption{\label{table:types} Global and Local Types}
\end{figure}
%
Figure~\ref{table:types} presents the syntax of the proposed global and local types. The type system assumes a set of ground types, ranged over by \g, including $\mathsf{int}$ and $\mathsf{bool}$. Type $\ti$ denotes the type of data-subject identities, whereas \hit\, 
denotes the anonymous identity type. The metasymbol \tid ranges over both $\ti$ and \hit.
An exchange type \U ranges over either ground types \g, or personal data store types of the form \antype{\annota}{\tid}{\g}.
Personal data store types \antype{\annota}{\tid}{\g} are
parameterised by an identity type \tid and a data type, \g, 
and are decorated with annotations, ranged over by \annota, \annotb, and \annotg. These annotations enable reasoning about the manipulation of private data through the corresponding database stores.

%


We assume that each data store is decorated with a static annotation $\annota$.
This annotation is part of the purpose definition and is not dynamically introduced by processes. 
%
%
We further assume that annotations uniquely determine the type of stored data.
That is, for a given annotation \annota, if a global type specifies two types
$\antype{\annota}{\tid}{\const[1]}$ and $\antype{\annota}{\tid}{\const[2]}$ then
$\const[1] = \const[2]$,


Global types \G define the interaction between the participants of a system.
Global type $\pass{\p}{\q}{\U} \G$ denotes the transmission of a value of type $\U$ from participant $\p$ to participant $\q$ and proceeding with global type $\G$.
Global type $\pass{\annota}{\p}{\pdtype{\tid}{\g}}\G$ denotes participant $\p$ reading from a store with annotation $\annota$ a value of type $\pdtype{\tid}{\g}$ before proceeding with \G.
Global type $\pass{\p}{\annota}{\pdtype{\tid}{\g}}$\G, denotes participant $\p$ storing a value of type $\pdtype{\tid}{\g}$ to personal data store annotated as \annota before proceeding with \G.
Global type $\choice{\p}{\q}{\lab[i] : \G_i}{i \in I}$ denotes participant $\p$ transmitting to participant $\q$ a label $\lab[i]$, and the protocol proceeding with the respective global type $\G_i$. Recursive protocols are modelled using $\mu t.\G$. Finally, global type $\ginact$ denotes the end of a global type.

\begin{example}[Simplified Medical Workflow Purpose]\label{ex3}
\rm
The following global types capture the purposes corresponding to the networks defined in Example~\ref{ex1}:
\vspace{-4pt}
\begin{eqnarray}\label{ex2}
    \G_1 &=& \pass{\patient}{\generalpract}{ {\antype{\annota}{\ti}{\symptoms}}}
            \pass{\annota}{\generalpract}{\ti \otimes \symptoms}
            \local
    \\
    \G_2 &=& \pass{\patient}{\generalpract}{{\antype{\annota}{\ti}{\symptoms}}}
            \pass{\annota}{\generalpract}{ \pdtype{\hid}{\symptoms}}
            \local
            \vspace{-3pt}
\end{eqnarray}
Global type $\G_1$ specifies that participant $\generalpract$ accesses personal data of type $\pdtype{\ti}{\symptoms}$ from a store decorated with annotation \annota, after receiving such store reference type from $\patient$. Similarly, global type $\G_2$ describes an anonymous personal data access. 
\qed
\end{example}
%
%
Local types, denoted by $T$, define communication interactions from the perspective of a single participant.
Local type $\tinp{\p}{\U}\,T$ denotes an input of a value of type \U from participant $\p$.
Local type $\tout{\p}{\U}\,T$ denotes an output of a value of type \U, towards participant $\p$.
Local type $\tinp{\annota}{\pdtype{\tid}{g}}\, T$ denotes the input of a personal data value of type $\pdtype{\tid}{\g}$ from a personal data store annotated with \annota. 
Dually, local type $\tout{\annota}{\pdtype{\tid}{\g}}\, T$ denotes the storing of a personal data value of type $\pdtype{\tid}{\g}$ to a personal data store annotated with \annota.
The \textit{select} local type $\tselect{\p}{\lab_i{:} T_i}_{i\in I}$ denotes the selection of a label $\lab_i$ from the offered set and its transmission to participant $\p$, before proceeding with type $T_i$.
Dually the \textit{branch} local type $\tbranch{\p}{\lab_i{:} T_i}_{i\in I}$, denotes the input of a label \lab[i] from participant $\p$ and proceeding according to type $T_i$. Recursion is modelled by the local type $\mu t.\local$, assuming equi-recursive type variable substitution, $\mu t. \local = \local \subst{\mu t.\local}{t}$.
Finally, local type $\tinact$ denotes the end of the local type.

\subsection{Local Projection of Global Types}
In this work we adopt the 
projection and merging operators from~\cite{ghilezan2019precise}, extended to handle the newly-introduced data-flow constructs:

\begin{definition}
\label{projectiondef}
The projection of a global type \G onto a participant $\q$, 
denoted as $\proj{\G}{\q}$, is a partial relation defined as follows,  where $\pts{\G}$ denotes the set of participants of $\G$:
\[
\setlength{\arraycolsep}{2pt}
\begin{array}{rcl}
	\proj{( \pass{\p}{\pd}{\U} \G)}{\q} &=&
	\left\{
		\begin{array}{ll}
			\tout{\pd}{\U}\proj{\G}{\q}  	&\,\;\quad\quad\quad\quad \text{if } \q = \p,
                \\[1mm]
				\tinp{\p}{\U}\proj{\G}{\q}       &\;\,\quad\quad\quad\quad \text{if } \q = \pd,
				\\[1mm]
			\proj{\G}{\q}                      &\;\,\quad\quad\quad\quad \text{otherwise}
		\end{array}
	\right.
	\\[9mm]
	\proj{( \pass{ \annota }{\p}{\pdtype{\tid}{\g}}\G)}{\q} &=&
	\left\{
		\begin{array}{ll}
          \tinp{\annota}{\pdtype{\tid}{\g}}
                \proj{\G}{\q}  		& \quad\quad\quad\text{if } \q = \p	
			\\[1mm]
			\proj{\G}{\q}          & \quad\quad\quad			\text{otherwise}
		\end{array}
	\right.
	\\[6mm]
	\proj{(\pass{\p}{\annota}{\pdtype{\tid}{\g}} \G)}{\q}\ &=&
	\left\{
		\begin{array}{ll}
        \tout{\annota}{\pdata{\tid}{\g}}
                \proj{\G}{\q}  		&\quad\quad\quad\text{if } \q = \p	
			\\[1mm]
			\proj{\G}{\q} &		\quad\quad\quad		\text{otherwise}
		\end{array}
	\right.
	\\[6mm]
	\proj{(\choice{\p}{\pd}{\lab[i] : \G_i}{i \in I}}){\q}\ &=&
	\left\{
		\begin{array}{ll}
			\tselecti{{\pd}}{\lab_i: \proj{\G[i]}{\q}}{i \in I}
  		    &\quad\text{if } \q = \p
			\\[1mm]
			\tbranchi{\p}{\lab_i: \proj{\G[i]}{\q}}{i \in I} 	&\quad\text{if } \q = \pd
			\\[1mm]
            \mrg_{i \in I}
			(\proj{\G_{i}}{\q}) &	\quad			\text{otherwise}
		\end{array}
	\right.
    \\[9mm]
        \proj{\ginact}{\q} &=& \ginact
\end{array}
\]
Above, $\mrg$ is a partial merging operator on local types, defined as follows:
\begin{eqnarray*}
	\local[1] \mrg \local[2] &=&
	\left\{
		\begin{array}{ll}
			\local[1]  	&\quad \text{if } \local[1] = \local[2],
            \\[1mm]
			\local[3]      &\quad \text{if }    \local[1] = \tbranchi{\p}{\lab_i: \local[i]}{i \in I},
                                                \local[2] = \tbranchi{\p}{\lab_j: \local[j]}{j \in J}
                                                \text{ and }
                                                \local[3] = \tbranchi{\p}{\lab_k: \local[k]}{k \in I \cup J},
		\end{array}
	\right.
\end{eqnarray*}
\qed
\end{definition}

Projection allows a participant to observe different actions in different branches of a choice even if that participant is not itself involved in the choice. This is essential for realistic workflows in which different branches involve different participants. For example, in the case study presented in Section~\ref{casestudy}, in a healthcare setting a laboratory may be involved only when a lab test is required, while in alternative branches, it does not participate at all.
\\
\begin{example}[Simplified Medical Workflow Projection]
\rm
As an example, projecting on \generalpract the global type defined in Example~\ref{ex3}, yields the following local type: 
\vspace{3pt}
\ceq{0}{\begin{array}{rcl}
        \proj{\G_1}{\generalpract} & = &    
        \tinp{\patient}{{\antype{\annota}{\tid}{\symptoms}}}
        \tinp{\annota}{\pdtype{\ti}{\symptoms}} \proj{\G_1}{\generalpract} 
    \end{array}
}
\vspace{3pt}
\noindent
 The \generalpract is receiving a store reference \annota from the \patient and proceeds with reading the \symptomsstore.
\qed
\end{example}

\subsection{Subtyping}

Subtyping, denoted by \(\subt\), is the largest relation between session
types coinductively defined by the rules in
Figure~\ref{fig:subtypingprocesses}, adapted from~\cite{ghilezan2019precise}. The subtyping relation preserves the usual safe substitutability principles for session types.
The relation is defined over value types, private data types, and local
session types. Base types, such as \(\integer\) and \(\bool\), are subtypes
only of themselves, as specified by rules \(\SUBINT\) and \(\SUBBOOL\).
Private-data types \(\pdtype{\tid}{\g}\) are covariant in the ground type
\(\g\) of the stored value.

Rule \(\SUBEND\) states that the inactive local type \(\tinact\) admits only
reflexive subtyping, i.e., it is a subtype only of itself.
Rule \(\SUBOUT\) defines output local types as covariant both in the
communicated type, \(\U \subt \Ud\), and in the continuation type,
\(\local \subt \locald\). Conversely, rule \(\SUBINP\) defines input local
types as contravariant in the communicated type, \(\Ud \subt \U\), and
covariant in the continuation type, \(\local \subt \locald\).

Rules \(\SUBPDOUT\) and \(\SUBPDINP\) extend the same principles to
private-data communication. Private-data output is covariant in the
private-data payload and in the continuation. Private-data input is
contravariant in the private-data payload and covariant in the continuation.

Rule \(\SUBSEL\) states that selection types are covariant and admit width
subtyping on the right: a selection offering fewer labels may be used where a
larger selection type is expected. Dually, rule \(\SUBBRA\) states that a
branch type accepting more labels may be used where a branch type accepting
fewer labels is expected. 
\begin{figure}[h]
\[
\begin{array}{c}
    \hline
    \\
    \SUBINT\; \integer \subt \integer
    \qquad
    \SUBBOOL\; \bool \subt \bool
    \qquad
    \dtree[\SUBDATA] {
        \g \subt \g'
    }{
        \pdtype{\tid}{\g} \subt \pdtype{\tid}{\g'}
    }
    \\[4mm]
    \SUBEND\; \tinact \subt \tinact
    \qquad
    \dtree[\SUBOUT] {
        \U \subt \U' 
        \quad
        \local \subt \local'
    }{
        \tout{\p}{\U}\local  \subt
        \tout{\p}{\U'}\local'
    }
    \qquad
    \dtree[\SUBINP] {
        \U' \subt \U 
        \quad \local \subt \local'
    }{
        \tinp{\p}{\U} \local \subt 
        \tinp{\p}{\U'} \local'
    }
    \\[4mm]
    \dtree[\SUBPDOUT] {
        \pdtype{\tid}{\g} \subt \pdtype{\tid'}{\g'}
        \quad
        \local \subt \local'
    }{
        \tout{\annota}{\pdtype{\tid}{\g}}  \local 
        \subt 
        \tout{\annota}{\pdtype{\tid'}{\g'}} \local'
    }
    \qquad
     \dtree[\SUBPDINP] {
        \pdtype{\tid'}{\g'} \subt \pdtype{\tid}{\g}
        \quad
        \local \subt \local'
    }{
        \tinp{\annota}{\pdtype{\tid}{\g}}  \local    
        \subt
        \tinp{\annota}{\pdtype{\tid'}{\g'}} \local'
    }
    \\[4mm]
    \dtree[\SUBSEL] {
        I \subseteq J \quad  \forall i \in I. T_i \subt T_i'
    }{
        \tselecti{\p}{\lab[i]: T_i}{i \in I}
        \subt
        \tselecti{\p}{\lab[j]: T'_j}{j \in J}
    }
    \qquad
    \dtree[\SUBBRA] {
         J \subseteq I \quad  \forall j \in J. T_j \subt T_j'
    }{
        \tbranchi{\p}{\lab[i]: T_i}{i \in I} 
        \subt
        \tbranchi{\p}{\lab[j]: T'_j}{j \in J}
    }
    \\
    \\
    \hline
\end{array}
\]
\caption{The subtyping relation}
\label{fig:subtypingprocesses}
\end{figure}

\section{
Type Checking Purpose Compliance}
\label{typesystem}
In this section, we present the type system of our calculus, which verifies whether a system implementation conforms to a declared purpose expressed as a global type 
 via a set of rules for checking whether a process or a network uses names and variables according to their types and follows the prescribed protocol associated with
the purpose.

The typing rules use {\em typing environment} \Ga,  
which associates references $u$ to personal data types of the form $\antype{\annota}{\tid}{\g}$, 
variables to ground types \g, placeholders $k$ to personal data types \pdtype{\tid}{\g}, and recursion variables to local types \local. 
\ceq{4} {
    \Ga \bnfis    \es \bnfbar 
                    \Ga, u: \antype{\annota}{\tid}{\g} \bnfbar 
                    \Ga, x: \g \bnfbar
                    \Ga, k: \pdtype{\tid}{\g} \bnfbar
                    \Ga, X: \local
}


Figure~\ref{fig:typesprocesses} defines the rules of the type system.
The first two lines specify the rules for asserting the types of store references, personal data variables, recursion variables, and value variables.
The next five lines define the typing rules for processes. Each rule has the form  $\Ga \proves \PP \as \local$, meaning that under environment \Ga process \PP has type \local.
The last line defines the typing rule for networks. 
\begin{figure}[h]
\[
\begin{array}{c}
    \hline
    \\
    \TData\; \Ga \proves \const: \g
    \qquad
    \TVdata\; \Ga \proves \id \as \ti
    \qquad
    \THdata\; \Ga \proves \hid \as \hit
    \qquad
    \TVar\; \Ga \cat x: \g \proves x: \g
    \\[2mm]
    \TPlace\; \Ga, k: \pdtype{\tid}{\g} \proves k: \pdtype{\tid}{\g}
    \qquad
    \TVarRec\; \Ga \cat X: \local \proves X \as \local
    \qquad
    \TPdata\; \Ga \cat u: \antype{\annota}{\tid}{\g} \proves u: \antype{\annota}{\tid}{\g}
    \\[6mm]
    \TInact \Ga \proves \inact \as \tinact
    \qquad
    \tree[\TOut] { 
        \Ga \proves t: \U
        \quad
        \Ga  \proves \PP \as \local 
    }{
        \Ga \proves \outp{\p}{t} \PP \as \tout{\p}{\U} \local
    }
    \qquad
    \tree[\TInp] {
        \Ga \cat x:\U \proves \PP \as \local
    }{
        \Ga \proves \inpp{\p}{x} \PP \as \tinp{\p}{\U} \local
    }
    \\[6mm]
    \tree[\TOutPD] {
        \Ga \proves u: \antype{\annota}{\ti}{\g} 
        \quad
        \Ga \proves \ii: \tid
        \quad
        \Ga \proves \const: \g
        \quad
        \Ga \proves \PP \as \local         
    }{
        \Ga \proves \outp{u}{\pdata{\ii}{\const}} \PP \as \tout{\annota}{\pdtype{{\tid}}{\g}} \local
    } 
    \\[6mm]
    \tree[\TInpPD] {
        \Ga \proves u: \antype{\annota}{\ti}{\g} 
        \quad
        \Ga, k: \pdtype{\tid}{\g} 
        \proves \PP \as \local 
    }{
        \Ga \proves \inpp{u}{k} \PP \as \tinp{\annota}{\pdtype{\tid}{\g}} \local
    }
    \qquad
    \tree[\TSel] {
        \Ga \proves \PP \as \local 
    }{
        \Ga \proves \selp{\p}{\lab} \PP \as \tselecti{\p}{\lab{:} T}{}
    }
        \\[6mm]
    \tree[\TBranch] {
        \forall i \in I, \Ga \proves \PP[i] \as \local[i]
    }{
        \Ga \proves \brapi{\p}{\lab[i]{:} P_i}{i\in I}\,\as\, \tbranchi{\p}{\lab[i]{:} T_i}{i \in I} 
    }   
        \qquad
     \tree[\TIf] {
       \Ga \proves \PP[1] \as \local 
        \qquad  \Ga \proves \PP[2] \as \local 
    }{
        \Ga \proves \ifelse{a}{b}{P_1}{P_2} 
        \as \local      
    }
     \\[6mm]
    \tree[\TRec] {
       \Ga, X : \local \proves \PP \as \local 
    }{
        \Ga \proves \mu X.P  \as \local
    }
      \qquad
     \tree[\TSub] {
       \Ga \proves \PP \as \local \qquad \local \leqslant \local'
    }{
        \Ga \proves \PP \as \local'
    }
    \\[6mm]
      \tree[\TNet] {
      \begin{array}{c}
         \forall i \in I, \Ga \proves \PP[i]:  \proj{\G}{\p[i]}
         \quad
         \pts{\G} \subseteq \set{\p[i]}_{i \in I}
         \\[1mm]
         \forall j \in J, \Ga \proves \rr[j]: \antype{\annota[j]}{\ti}{\g[j]} \land \Ga \proves \const[j]: \g[j]
         \quad
         \rr[j]: \antype{\annota[j]}{\ti}{\g[j]} \in \Ga \implies j \in J
         \end{array}
    }{
       \Ga \proves \prod_{i \in I}  \proc{\p[i]}{\PP[i]} \Par \prod_{j\in J}\store{\rr[j]}{\ti}{\const[j]} \,\as\, \G
    }
    \\
    \\
    \hline
\end{array}
\]
\caption{Typing rules for Constants, References, Variables, Private data, Processes, and Network}
\label{fig:typesprocesses}
\end{figure}

Starting with the rules for typing constants, references, and variables:
Rule \TData assigns a constant \const to a ground type \g, based on the type of the constant (e.g., integer, boolean).
Similarly, rules \TVdata and  \THdata assign the types \ti 
and \hit\, to non-anonymous identities and the anonymous identity, respectively.
Rule \TVar types variables with ground types, according to \Ga.
Subsequently, rules \TPlace,  \TVarRec, and \TPdata type
personal data variables $k$, recursion variables $X$, and
store terms $u$ according to \Ga.

Moving on to processes, typing rule \TInact types the inactive process with local type $\tinact$.
Typing rule \TOut types processes that send a value of type \U to participant \p with local type $\tout{\p}{\U} \local$, where \local is the type of the continuation process. Respectively, typing rule \TInp types processes that read a value of type \U from participant \p with local type $\tinp{\p}{\U}\local$.
%
%
Typing rule \TOutPD types a process that writes a personal data value of type $\pdtype{\tid}{\g}$ to a store via a reference of type \antype{\annota}{\ti}{\g}, with type $\tout{\annota}{\pdtype{\tid}{\g}}\local$.
Similarly,
typing rule \TInpPD types a process that inputs a value of type $\pdtype{\tid}{\g}$ via a store reference of type  \antype{\annota}{\ti}{\g}, with local type $\tinp{\annota}{\pdtype{\tid}{\g}} \local$. In both cases, \local
is the type of the continuation processes. Furthermore,
the use of $\ti$ and $\tid$, in the value
and store types of the rule, accounts to the fact that data manipulation
may be anonymized.
Typing rule \TSel types a process that sends a label \lab[] to participant \p with local type $\tselecti{\p}{\lab{:} T}{}$.
Typing rule \TBranch types a process that receives and matches a label from a choice of labels $\lab[i], i \in I$, from participant \p, with local type $\tbranchi{\p}{\lab[i]: \local[i]}{i \in I}$.
Finally, rule \TSub captures subsumption: if $P$ is well-typed with type $T$ and $T$ is a subtype of $T'$, then $P$ is also well-typed with type $T'$.

Moving to the typing rule for networks, typing rule \TNet types a network to a global type \G assuming that the participant processes composed in parallel are well-typed according to their local type projection of \G, all store processes
in the network are well-typed according to $\Gamma$, and all stores
in $\Gamma$ are present in the network. Moreover, the requirement that $\pts{\G} \subseteq \set{\p[i]}_{i \in I}$,
allows to type networks with participants of the form
$\proc{\p}{\inact}$, thus ensuring that typing remains invariant under structural congruence.


\section{Subject Reduction and Purpose Fidelity}
\label{theorems}
We establish 
\emph{subject reduction} and \emph{purpose fidelity} for our framework. We first present the Subject reduction Theorem (Theorem~\ref{subjredtheorem}), stating that the execution of a typed network, according to the operational semantics, preserves the overall typing of the network. 
%
%
Moreover, the Purpose Fidelity Theorem (Theorem~\ref{purpfidelitytheorem}) states that the communication sequence of a process follows the purpose declared.
In order to state the Subject Reduction Theorem we need to 
define how global types are modified during the evolution of networks.

\begin{definition}[Global type consumption and reduction] 
\label{consreddef}
The {\em consumption} of an action $\act$ in a global type \G 
(notation ${\G}\; \backslash \, {\act}$) is defined
co-inductively as follows.
\[
\begin{array}{c}
    \setlength{\arraycolsep}{2pt}
    \begin{array}{rclcrcccl}
    \act	&\bnfis&  \gconsdef{\p}{\q}{\U}
            \bnfbar \gconsdef{\p}{\q}{\lab[k]} &\hspace{3cm}& \pts{\gconsdef{\p}{\q}{\U}} &=& \pts{\gconsdef{\p}{\q}{\lab[k]}} &=& \set{\p, \q}
    \\      &\bnfbar& \gconsdef{\annota}{\p}{\pdtype{\tid}{\g}}
            \bnfbar \gconsdef{\p}{\annota}{\pdtype{\tid}{\g}} && \pts{\gconsdef{\annota}{\p}{\pdtype{\tid}{\g}}} &=& \pts{\gconsdef{\p}{\annota}{\pdtype{\tid}{\g}}} &=& \set{\p, \annota}
    \end{array}
    \\[6mm]
    \gcons{\pass{\p}{\q}{\U}\G}{\p}{\q}{\U} = \G
    \qquad
    \dtree {
        \gconsShort{\G}{\act} = \G'
        \quad
        \set{\p[1],\q[1]} \cap \pts{\act} = \es
    }{
        \gconsShort{\pass{\p[1]}{\q[1]}{\U[1]}\G
        }{
            \act
        } = \pass{\p[1]}{\q[1]}{\U[1]}\G'
    } 
    \\[4mm]
    \gcons{
        \choice{\p}{\q}{\lab[i] : \G_i}{i \in I}
    }{\p}{\q}{\lab[k]} = \G_k \;\; k \in I
    \qquad
    \dtree{
        \forall i \in I, \gconsShort{\G_i}{\act} = \G_i'
        \quad
        \set{\p[1],\q[1]} \cap \pts{\act} = \es
    }{
        \gconsShort{
            \choice{\p[1]}{\q[1]}{\lab[i] : \G_i}{i \in I}
        }{\act} = \choice{\p[1]}{\q[1]}{\lab[i] : \G_i'}{i \in I}
    }
    \\[4mm]
    \gcons{\pass{\annota}{\p}{\pdata{\tid}{\g}}\G}{\annota}{\p}{\pdtype{\tid}{\g}} = \G
    \qquad
    \dtree{
        \gconsShort{\G}{\act} = \G'
        \quad
        \set{\annota,\p} \cap \pts{\act} = \es
    }{
        \gconsShort{ \pass{\annota}{\p}{\pdtype{\tid}{\g}} \G }{\act} = \pass{\annota}{\p}{\pdtype{\tid}{\g}} \G'
    }
    \\[4mm]
    \gcons{\pass{\p}{ \annota }{\pdata{\tid}{\g}} \G}{\p}{\annota}{\pdata{\tid}{\g}} = \G
    \qquad
    \dtree{
        \gconsShort{\G}{\act} = \G'
        \quad
        \set{\p, \annota} \cap \pts{\act} = \es
    }{
        \gconsShort{ \pass{\p}{\annota}{\pdtype{\tid}{\g}} \G }{\act} = \pass{\p}{\annota}{\pdtype{\tid}{\g}} \G'
    }
\end{array}
\]
Assume a global type \G. The {\em reduction} of $\G$ is defined co-inductively as the smallest pre-order 
relation closed under the rule
$\G \trans{} \G'$, where $\gconsShort{\G}{\act} = \G'$ for
some $\G'$ and $\act$.
\end{definition}

\begin{theorem}[Subject Reduction] 
\label{subjredtheorem}
    Assume a 
    network \MM such that $\Ga \proves \MM \as \G$. 
    If \MM executes a computation step \MM \trans{} \MMd, then the resulting system \MMd is also
    typed with $\Ga \proves \MMd \as \G'$ such that $\G \trans{} \G'$.
\end{theorem}

\begin{proof}
    The proof is done by induction on the definition of the operational semantics relation. The details of the proof can be found in
    ~\cite{vanezigdprmst}. 
\end{proof}

\begin{theorem}[Purpose Fidelity]
\label{purpfidelitytheorem}
    Assume $\Ga \proves \MM \as \G$.
    \begin{itemize}
    \setlength{\itemsep}{2pt}
\setlength{\topsep}{2pt}
        \item   If $\G = \inact$ then $\MM \scong \inact$;
        \item   If there exists $\G'$ such that $\gconsShort{\G}{\gconsdef{\p}{\q}{\U}} = \G'$ then there exists \MMd, $\G''$ such that $\MM \trans{} \MMd$ with $\Ga \proves \MMd \as \G''$ and either $\G' = \G$ or $\G' = \G''$;
        \item   If there exists $\G'$ such that $\gconsShort{\G}{\gconsdef{\p}{\annota}{\pdtype{\tid}{\g}}} = \G'$ then there exists \MMd, $\G''$ such that $\MM \trans{} \MMd$ with $\Ga \proves \MMd \as \G''$ and either $\G' = \G$ or $\G' = \G''$;
        \item   If there exists $\G'$ such that $\gconsShort{\G}{\gconsdef{\annota}{\p}{\pdtype{\tid}{\g}}} = \G'$ then there exists \MMd, $\G''$ such that $\MM \trans{} \MMd$ with $\Ga \proves \MMd \as \G''$ and either $\G' = \G$ or $\G' = \G''$;
        \item   If there exists $\G''$ such that $\gconsShort{\G}{\gconsdef{\p}{\q}{\lab[k]}} = \G''$ then there exists $\G', \lab[m]$ such that $\gconsShort{\G}{\gconsdef{\p}{\q}{\lab[m]}} = \G'$, and there exists \MMd, $\G'''$ such that $\MM \trans{} \MMd$ with $\Ga \proves \MMd \as \G'''$ and either $\G' = \G$ or $\G' = \G'''$.
    \end{itemize}
\end{theorem}

\begin{proof}
We note that in the last four clauses $\G=\G'$ may arise due
to conditional statements, whereas, in the last clause, $\G'''$ may arise due to subtyping. The details of the proof can be found in
~\cite{vanezigdprmst}. 
\end{proof}

Intuitively, these two results establish our goal of
assuring that an implementation is purpose compliant. In particular, recall that a multiparty purpose protocol (a global type) in
our framework provides a structured specification of a privacy purpose.
Such a protocol describes how multiple participants exchange information and coordinate control over private data and private data stores (e.g., databases).
Assume a program \MM handling private data adheres to a purpose \G under an environment \Ga, i.e.,
\vspace{2pt}
\ceq{0}{
    \Ga\proves\MM\as\G.
}
\vspace{2pt}
The Subject Reduction Theorem ensures that all execution steps of \MM conform to the purpose \G: every interaction performed by \MM is prescribed by \G in the sense that after  evolution
\MM remains well-typed by a global type to which $\G$ can reduce. In other words, the program cannot exhibit behaviour outside the specified protocol.
Conversely, the Purpose Fidelity Theorem guarantees that \MM is never stuck when the purpose \G prescribes an interaction: whenever \G can consume an action, the program can perform a corresponding computation step.
Together, these results establish a tight correspondence between programs and purposes:
\begin{description}
    \item[Soundness] (Subject Reduction): programs do not exceed the specified purpose.
    \item[Completeness] (Purpose Fidelity): programs can realise all interactions required by the purpose.
\end{description}
Hence, the theory provides strong guarantees that purpose specifications are both safe (no unintended behaviour) and enforceable (all intended behaviour is executable) for programs handling private data.

\section{Case Study: Medical Diagnostics Workflow}
\label{casestudy}
To illustrate the expressiveness of our calculus and the applicability of our typing and compliance framework, we present a case study based on a medical diagnostics workflow, inspired by common healthcare scenarios in the formal privacy and purpose-limitation literature~\cite{TschantzDW12,jafari2014framework,jafari2009enforcing}. 
The case study demonstrates how intended processing purposes are specified, how system implementations are modelled, and how both compliance and violations are detected by our type system.

\subsection{Case Study Overview and Intended Purpose}
The workflow concerns a patient undergoing a diagnostic process managed by a digital healthcare system. 
The diagnostic process unfolds as follows: The patient (participant \patient) initially sends their personal information (via store $\rr_{i}$), writes and sends their symptoms (via store $\rr_{s}$), both to be received by a nurse practitioner (participant \generalpract). The nurse reads the
data from the two stores and performs an initial screening, determining whether further laboratory testing is required before forwarding the case to the general practitioner (participant \diagnosiseng). Depending on the decision, if a lab test is needed,
the nurse creates and sends a laboratory order (store $\rr_{o}$) to the lab (participant \laboratory), who reads the order, performs the test, writes down the results (store $\rr_{r}$), and sends them to the nurse. In both cases, the nurse sends the symptoms (store $\rr_{s}$) to the general practitioner, and if there was a lab test also the results (store $\rr_{r}$). The general practitioner then reads the received data, performs a diagnostic assessment, and records the diagnosis (store $\rr_{d}$), which is returned to the nurse. The nurse then forwards the diagnosis (store $\rr_{d}$) to the patient. Throughout this workflow, access to personal data is restricted to what is necessary for performing the diagnostic process, in accordance with the intended purpose.

\subsection{System Participants and Data Stores}
The workflow involves the following participants: patient (\patient), nurse practitioner (\generalpract), laboratory (\laboratory), and general practitioner (\diagnosiseng). 
As discussed in Section~\ref{calculus}, personal data are stored in dedicated data stores, which are accessed via references. 
The following store references and types appear: 
\begin{itemize}
\setlength{\itemsep}{-1pt}
    \item $\rr_{i}$ (demographic and administrative information e.g., age, patient ID), of type $\antype{\infostore}{\id}{\basicinfo}$, i.e.,
    decorated with the static annotation $\infostore$ and holding data of type $(\id \otimes \basicinfo)$;
    \item $\rr_{s}$ (symptoms reported by the patient), with type 
    $\antype{$\symptomsstore$}{\id}{\symptoms}$
    %
    \item $\rr_{o}$ (laboratory test order), with type $\antype{$\lorderstore$}{\id}{\laborder}$
    %
    \item $\rr_{r}$ (laboratory results), with type 
    $\antype{$\resultsstore$}{\id}{\labresults}$
    %
    \item $\rr_{d}$ (final diagnostic report), with type
    $\antype{$\diagnstore$}{\id}{\diagnosticreport}$
    %
\end{itemize}


\subsection{Formal Specification of Processing Purposes}
The diagnostic processing purpose of the workflow is formalised as the following global type ($\G_{diagnostic}$), which specifies the allowed interaction patterns and data accesses between participants in compliance with the intended diagnostics purpose. 
For readability reasons, we factor the final diagnostic continuations into `withLab' and `noLab'.
The specification distinguishes between diagnostic paths that require laboratory testing and those that do not, while ensuring that personal data are accessed only by specified entities and only when required by the diagnostic purpose.
For brevity, when exchanging store references, we write only the annotation, omitting the full type info. For example, we write $\pass{\patient}{\generalpract}{\infostore} \G$ instead of 
$\pass{\patient}{\generalpract}{\infostore\privatet{\pdata{\id}{\basicinfo}}^{}} \G$.

\[
\begin{array}{rcl}
\G_{diagnostic}&=&
\pass{\patient}{\symptomsstore}{\id \otimes \symptoms}
    \pass{\patient}{\generalpract}{\infostore} \\ &&
    \pass{\patient}{\generalpract}{\symptomsstore} 
    \pass{\infostore}{\generalpract}{\id \otimes \basicinfo} \\ &&
    \pass{\symptomsstore}{\generalpract}{\id \otimes \symptoms} \\ &&
         \choicetable{\generalpract}{\laboratory}{ 
         \labno: 
\text{noLab}, 
            \\
            \labyes: \text{withLab} 
            }
\end{array}
\]
\[
\begin{array}{rcl}
\text{noLab} &=&
    \choicetable{\generalpract}{\diagnosiseng}{
            \labno: & \pass{\generalpract}{\diagnosiseng}{\symptomsstore} \\ &
            \pass{\symptomsstore}{\diagnosiseng}{\id \otimes \symptoms} \\ &
    \pass{\diagnosiseng}{\diagnstore}{\id \otimes \diagnosticreport}\\ &
    \pass{\diagnosiseng}{\generalpract}{\diagnstore} \\ &
    \pass{\generalpract}{\patient}{\diagnstore}\\& \ginact
    }
    \end{array}
\]
\[
\begin{array}{rcl}
  \\[5mm]
   \text{withLab} &=&
   \pass{\generalpract}{\lorderstore}{\id \otimes \laborder}  
   \pass{\generalpract}{\laboratory}{\lorderstore} 
    \\&&  \pass{\lorderstore}{\laboratory}{\id \otimes \laborder} 
    \pass{\laboratory}{\resultsstore}{\id \otimes \labresults}\\ &&
     \pass{\laboratory}{\generalpract}{\resultsstore} \\&&
      \choicetable{\generalpract}{\diagnosiseng}{
            \labyes: &  \pass{\generalpract}{\diagnosiseng}{\symptomsstore}\\ &
            \pass{\generalpract}{\diagnosiseng}{\resultsstore} \\&
    \pass{\symptomsstore}{\diagnosiseng}{\id \otimes \symptoms} \\&
    \pass{\resultsstore}{\diagnosiseng}{\id \otimes \labresults} \\&
    \pass{\diagnosiseng}{\diagnstore}{\id \otimes \diagnosticreport} \\&
    \pass{\diagnosiseng}{\generalpract}{\diagnstore} \\&
    \pass{\generalpract}{\patient}{\diagnstore}\\ &\ginact}
    \\[10mm]
\end{array}
\]

We now project the global type, into local types for each participant. We demonstrate below the local type for participant \diagnosiseng. The projection reflects that \diagnosiseng receives a label from \generalpract, and subsequently receives and accesses the corresponding data depending on the diagnostic path, and performs a diagnosis. 

\[
    \begin{array}{rcl}
        \proj{\G}{\diagnosiseng} & = & 
        \tchoicetable{\generalpract}{ 
        \labyes: &
         \tinp{\generalpract}{\symptomsstore}
         \tinp{\generalpract}{\resultsstore}\\ &
         \inppd{\symptomsstore}{\id \otimes \symptoms}
         \inppd{\resultsstore}{\id \otimes \labresults}\\&
          \outpd{\diagnstore}{\id \otimes \diagnosticreport}
          \tout{\generalpract}{\diagnstore} \tinact
         ,\\
        \labno: &
        \tinp{\generalpract}{\symptomsstore}\\ &
        \inppd{\symptomsstore}{\id \otimes \symptoms}
        \outpd{\diagnstore}{\id \otimes \diagnosticreport}\\ &
        \tout{\generalpract}{\diagnstore}
        \tinact }
    \end{array}
\]
The complete set of the projected local types can be found in
~\cite{vanezigdprmst}. 

\subsection{Modelling a Purpose-Compliant Implementation}

We 
use the calculus syntax to
model a system implementation that follows the interaction prescribed by the diagnostic purpose global type, and is thus intended to be type-compliant with $\G_{diagnostic}$. We present the process for the \diagnosiseng, as well as the network. All other participants are defined analogously and conform to their projected local types and can be found in
~\cite{vanezigdprmst}.
\[
	\begin{array}{rcl}  
              D &=&
                \bratable{\generalpract}{
                    \labyes: &  \inpp{\generalpract}{x_{s}} 
                \inpp{\generalpract}{x_{r}}
                \inpp{x_{r}}{x_{id1} \otimes z_{r}} 
                \inpp{x_{s}}{x_{id1} \otimes z_{s}} \\ &
                \outp{r_{d}}{x_{id1} \otimes diagn_1}
                \outp{\generalpract}{r_{d}} 
                \inact, \\
                \labno: &       \inpp{\generalpract}{x_{s}}
                \inpp{x_{s}}{\id \otimes z_{s}}
                \outp{\rr_{d}}{\id \otimes diagn_1} \outp{\generalpract}{r_{d}} \inact
                } 
            \\ \\
             \MM & = & \proc{\generalpract}{A} \Par \proc{\patient}{B} \Par \proc{\laboratory}{C} \Par 
            \proc{\diagnosiseng}{D}  \\&&
            \Par \store{\rr_{i}}{\id_1}{\basicinfo}
            \Par \store{\rr_{s}}{\id_1}{w_1} 
            \Par \store{\rr_{o}}{\id_1}{w_2} \\&&
            \Par \store{\rr_{r}}{\id_1}{w_3} 
            \Par \store{\rr_{d}}{\id_1}{w_4}
    \end{array}
\]

\paragraph{Compliance Verification via Type Checking.}
Using our type system, we may check whether the system's behaviour complies with the specified purpose. Demo application of the typing rules can be found in
~\cite{vanezigdprmst}. 
By applying the typing rules, we observe that the process of each participant corresponds to the associated projected local type from $\G_{diagnostic}$. 
Consequently, network $M$ is typable under the declared global type $\G_{diagnostic}$. Therefore, $M$ is purpose-compliant.


           

\subsection{Modelling and Detection of Non-Compliant Behaviour}
We introduce a deviation of the model in which the general practitioner (\diagnosiseng) accesses the patient’s medical data without completing the diagnostic process, resulting in processing that does not fulfil the declared purpose.
\[
	\begin{array}{rcl}
             D_{nc} &=&
                \bratable{\generalpract}{
                    \labyes: &  \inpp{\generalpract}{x_{s}} 
                \inpp{\generalpract}{x_{r}}
                \inpp{x_{r}}{\id \otimes z_{r}}
                \inpp{x_{s}}{\id \otimes z_{s}} 
                \inact, \\
    \labno: &   \inpp{\generalpract}{x_{s}}
                \inpp{x_{s}}{\id \otimes z_{s}}
                \inact
                } 
                \\ \\
             M_{nc} & = & \proc{\generalpract}{A} \Par \proc{\patient}{B} \Par \proc{\laboratory}{C} \Par 
            \proc{\diagnosiseng}{D_{nc}}  \\&&
            \Par \store{\rr_{i}}{\id_1}{\basicinfo}
            \Par \store{\rr_{s}}{\id_1}{w_1} 
            \Par \store{\rr_{o}}{\id_1}{w_2} \\&&
            \Par \store{\rr_{r}}{\id_1}{w_3} 
            \Par \store{\rr_{d}}{\id_1}{w_4}
    \end{array}
\]

\paragraph{Violation Detection via Type Checking.}
In this case, application of the typing  rules indicates that all participants except \diagnosiseng correspond to their roles in the projected local types. 
The type system detects this unauthorised interaction, and the network $M_{nc}$ is rejected as non-compliant.

\section{Conclusions}
\label{concl}
Aiming to address the GDPR \emph{purpose limitation} principle which mandates that systems should collect and process personal data according to explicitly specified purposes, and building upon our previous work of~\cite{EVpurpose,vanezi2025privacy} and relevant literature, we have presented a formal, purpose-aware framework grounded in multiparty session types. The presented approach allows the specification of purposes in systems as structured interaction protocols, the modelling of system implementations that process personal data using a process calculus, and the compliance verification between the two, recognising and reporting non-compliant system implementations through a type system. 
To achieve this, the framework integrates the notion of \emph{personal data}~\cite{KP-LMCS17} into multiparty session types~\cite{takeuchi1994interaction,honda1998language,honda2008multiparty}.
We establish subject reduction and purpose fidelity results.
%
%
Our approach bears promise for practical applicability in software engineering through the direct correspondence between the proposed global types for capturing purposes and purpose-enhanced UML sequence diagrams~\cite{EVpurpose,vanezi2025privacy}.

As future work, we plan to extend the type system with a linear dimension to ensure properties such as uniqueness of store references. We also aim to support additional analyses, including model checking, to validate further properties of purposes and other GDPR requirements.
%
Moreover, we plan to extend the calculus to support more complex system definitions, such as multiple parallel sessions and dynamic store creation. 
Ultimately, our objective is to evolve this work 
into a software-engineering-oriented approach that unifies purpose modelling and compliance verification within a lifecycle-driven methodology, enabling end-to-end purpose compliance validation. As such, we envision the development of a tool that will automate the  formal purpose-compliance validation process, and support its integration into software engineering practice, alongside sequence diagrams, thereby advancing a practical privacy-by-design approach.

\newpage
\nocite{*}
\bibliographystyle{eptcs}
\bibliography{generic}
\end{document}